\documentclass[]{article} 
\usepackage{mathtools}
\usepackage{float}
\usepackage[utf8]{inputenc}
\usepackage{verbatim}
\usepackage{smartdiagram}
\usepackage{todonotes}
\usepackage{caption}
\usepackage{amsthm}

\usepackage{graphicx}
\usepackage{multicol}
\usepackage{footmisc}
\usepackage{cite}
\usepackage{url}

\usepackage{amsmath, amssymb, amsfonts}

\usepackage[linesnumbered, boxed, lined, ruled,vlined]{algorithm2e}

\newtheorem{definition}{Definition}[section]
\newtheorem{theorem}[definition]{Theorem}
\newtheorem{lemma}[definition]{Lemma}
\newtheorem{corollary}[definition]{Corollary}

\newtheorem{example}[definition]{Example}

\newcommand{\N}{\mathbb{N}}

\newcommand{\R}{\mathbb{R}}

\title{On the Connection between Individual Scaled Vickrey Payments and the Egalitarian Allocation}

\author{N. Gräf \and T. Heller\footnote{\texttt{till.heller@itwm.fraunhofer.de}, corresponding author} \and S.O. Krumke}

\begin{document}         

	\maketitle              
	
	\begin{abstract}
The Egalitarian Allocation (EA) is a well-known profit sharing method for cooperative games which attempts to distribute profit among participants in a most equal way while respecting the individual contributions to the obtained profit. Despite having desirable properties from the viewpoint of game theory like being contained in the core, the EA is in general hard to compute. Another well-known method is given by Vickrey Payments (VP). Again, the VP have desirable properties like coalitional rationality, the VP do not fulfill budget balance in general and, thus, are not contained in the core in general. One attempt to overcome this shortcoming is to scale down the VP. This can be done by a unique scaling factor, or, by individual scaling factors. Now, the individual scaled Vickrey Payments (ISV) are computed by maximizing the scaling factors lexicographically. In this paper we show that the ISV payments are in fact identical to a weighted EA, thus exhibiting an interesting connection between EA and VP. With this, we conclude the uniqueness of the ISV payments and provide a polynomial time algorithm for computing a special weighted EA. 
\end{abstract}

\section{Introduction}
Cooperative game theory deals with the question of how a cooperation of participants can be achieved by incentives such that at the same time a possible turning away from the cooperation is seen as not appealing. For this, the question of how a jointly achieved profit (or cost) is to be distributed is central. It is important to consider both individual contributions to profit and the gains achieved by a group. 

Well established fairness properties are the following three. First of all, \emph{individual rationality (IR)} ensures that a participant cannot incur a loss by participating. The \emph{coalitional rationality (CR)} is a natural extension to a group of participants, i.e. a group of participants is equal or better off by participating. The third one ensures that the profit shared should be covered by the generated profit and is called \emph{budget balance(BB)}. 

A profit distribution that fulfills these desired properties is the \emph{egalitarian allocation (EA)} which was first introduced by Dutta and Ray (cf. \cite{dutta1989concept}). The egalitarian allocation shares the profit generated by a coalition of participants, i.e. a group of participants, in an equal way, starting from the coalition with the highest generated profit within themselves. 

Koster \cite{koster1999weighted} extended the egalitarian solution to asymmetric cases by introducing weights to all participants. Both the uniqueness and existence results from~\cite{dutta1989concept} were obtained in the asymmetric case. In order to compute the egalitarian allocation, \cite{koster1999weighted} proposed an algorithm that requires to compute a maximally weighted coalition that maximizes the weighted average contribution in every iteration. In general, this needs an exponential number of evaluations of the characteristic function and, thus, the algorithm has an exponential running time in general. 

Another important and famous solution concept are the \emph{Vickrey payments}, which are based on work from Clarke \cite{clarke1971multipart}, Groves \cite{groves1973incentives} and Vickrey \cite{vickrey1961counterspeculation}. Despite the properties (e.g. IR, CR) of Vickrey payments, a major drawback is that the Vickrey payments do not fulfill BB. In order to get a profit distribution that fulfills the budget balance constraint, the Vickrey payments can be scaled down in several ways. Parkes et al. \cite{parkes2001achieving} introduced the \emph{threshold Vickrey payments}, where discounts are given to the participant only if they exceed some given threshold. Another way is to scale down all Vickrey payments by the same rate, such that the budget balance constraint is satisfied. This was analyzed also by Parkes et al. \cite{parkes2001achieving} and by Ackermann et al. \cite{ackermann2014modeling}. By scaling down the Vickrey payments with respect to their relative weights, Ackermann et al. \cite{ackermann2014modeling} showed that they obtained a polynomial time algorithm for a profit distribution that fulfills IR, CR and BB.

The outline of the paper is as follows. In the following section, we introduce the basic notation of cooperative game theory. In Section~\ref{sec: egal: vickrey} we recall the definition of Vickrey payments and using examples to show basic properties and shortcomings whereas we present a short introduction to the egalitarian allocation in Section~\ref{sec: egal: egalitarian}. Our main contribution is to prove the connection between the weighted egalitarian allocation (WEA) and the individually scaled Vickrey payments (ISV), which is presented in Section~\ref{sec: egal: connection}. We then conclude with a short outlook.

\section{Preliminaries and basic examples}
First, we start with the basic notation of cooperative game theory. A \emph{cooperative game with transferable utility} (also denoted by \emph{TU-game}) is a tuple~$(N,v)$, where $N$ denotes a set of \emph{participants} and $v: 2^N\rightarrow \R$ denotes the \emph{characteristic function}. The characteristic function~$v$ assigns every subset of participants a value in $\R$. A (sub)set of participants is called a \emph{coalition} whereas we refer to the coalition consisting of all participants as the \emph{grand coalition} $N$.  
A \emph{payoff vector} is a vector in $\R^{|N|}$ whose $i$-th entry denotes the payment to participant~$i$. If a payment vector~$p$ is \emph{feasible}, i.e. $\sum_{i\in N} p(i) = v(N)$, and \emph{individually rational}, i.e. $p(i)\geq v(\{i\})$ for all~$i\in N$, we call $p$ an \emph{imputation}. The \emph{core} is defined as the set of all imputations that furthermore fulfill \emph{coalitional rationality}, i.e. $\sum_{i\in S} x_i \geq v(S)$.

\subsection{Some Notes on Vickrey Payments}\label{sec: egal: vickrey}
The \emph{Vickrey payments} assign each participant a payoff equal to the value generated by the participation of the respective participant. For this let $N$ denote the set of participants and $v(N)$ be the evaluation of the grand coalition~$N$ of the characteristic function~$v$. Then the Vickrey payment~$vp_i$ of a participant~$i$ is given by
\begin{align*}
vp_i \coloneqq v(N) - v(N\backslash \{i\}).
\end{align*}

Ausubel et al. \cite{ausubel2002ascending} gave a direct connection between the core elements and the Vickrey payments.
\begin{theorem}[cf. \cite{ausubel2002ascending}]\label{thm: vickrey core}
	The Vickrey payment of a participant is equal to the participant's highest payoff over all core elements. $\Box$
\end{theorem}

Unfortunately, the payoff vector consisting of the participants' Vickrey payments does not lie in the core in general as the following example shows.
\begin{example}[Violation of BB]
	We define a cooperative game depending on the setting of a combinatorial exchange as follows. The set of participants is given by the union of a set of supply bids and a set of demand bids for a kind of item. The characteristic function now maps each coalition of bids to the obtained profit, i.e. the difference between the demand offer and the supply offer, in said coalition. Suppose we are given an item~$I$, one supply bid~$s_1$ that offers item~$I$ for a price of $10$, and two demand bids~$d_1, d_2$ that are willing to pay $12$, resp. $15$. The total auction profit is therefore $5$, but the Vickrey payments are $vp_{s_1} = 5$, $vp_{d_1} = 5 - 2 = 3$ and $vp_{d_2} = 0$. Thus, the sum of all Vickrey payments is higher than the total auction profit.
\end{example}

Parkes et al. \cite{parkes2001achieving} came up with several strategies to regain budget balance. For example, the \emph{threshold Vickrey payments} are essentially the Vickrey payments, but only consider the payments that are larger than a given threshold. By choosing a suitable threshold, one can enforce that not more than the budget is spent on payments. Another technique they presented is the (equally) scaling approach, where all the Vickrey payments are scaled by a factor in order to regain budget balance. For this let $\alpha$ denote the scaling factor, $P$ as the budget and $vp_i$ denote the Vickrey payment for participant~$i$. Then the \emph{(equally) scaled Vickrey payment}~$vp_i^{eq}$ (ESV) of participant~$i$ can be computed by solving the following mathematical program.
\begin{align*}
\max \quad & \alpha\\
\text{s.t.} \quad & vp_i^{eq} = \alpha \cdot vp_i \qquad \forall i\in N \\
& \sum_{i\in N} vp_i^{eq} = P\\
& \alpha \in [0, 1].
\end{align*}

Unfortunately, the next example shows that such a payoff vector does not always fulfill coalitional rationality, i.e. does not always lie in the core.

\begin{example}[Violation of CR]
	We extend the example above by adding another supply bid~$s_2$ with a price of $12$ and another demand bid~$d_3$ with price~$13$. Now the value of the grand coalition is given by $v(N) = (15-10) + (13-12) = 6$. The Vickrey payments now are given by $vp_{s_1} = 3 = vp_{d_2}$, $vp_{s_2} = 1 = vp_{d_3}$ and $vp_{d_1} = 0$. As they sum up to $8$, the scaling factor is $\alpha= \frac{6}{8}$. But now the coalition $\{s_1, d_2\}$ could obtain a profit of $5$ on their own, but is only rewarded by $3\cdot \frac{6}{8} \cdot 2 = \frac{9}{2} < 5$ and thus does not satisfy CR. 
\end{example}

Ackermann et al. showed in \cite{ackermann2014modeling} that scaling the Vickrey payments on an individual basis keeps coalitional rationality in their setting. More precisely, they introduced a scaling factor for each entry of the payment vector and choose the lexicographical maximal scaling vector. Using the same notation as before, the \emph{individually scaled Vickrey payments}~$vp^{ind}$ (ISV) can be computed by solving the following mathematical program where we maximize lexicographically, i.e. the smallest entry of the vector~$\alpha$ is maximized.
\begin{align}
\max\phantom{}^{\prec} \quad & \alpha \label{eq: obj func isv}\\
\text{s.t.} \quad & vp_i^{ind} = \alpha_i \cdot vp_i \qquad \forall i\in N \notag \\
& \sum_{i\in N} vp_i^{ind} = P \notag\\
& \alpha_i \in [0, 1] \qquad \forall i\in N. \notag
\end{align}

\subsection{Some Notes on the Egalitarian Allocation}\label{sec: egal: egalitarian}
Before we state the definition of the \emph{(weighted) egalitarian allocation}, we introduce some notation. A permutation, that orders the entries $x^w_i\coloneqq \frac{x_i}{w_i}$ of an entry-wise scaled vector~$x^w$ with a weight vector~$w$, is denoted by~$\sigma_x^w$. First, we recall some basic definitions regarding the egalitarian allocation:
\begin{definition}[weighted Lorenz-curve, Lorenz-domination]
	Given a vector~$x$ in $\R^n$, weights~$w$ in $\R^n_{> 0}$ and $W\coloneqq \sum_{i=1}^n w_i$. Let $\sigma^w_x:\N \rightarrow \N$ be a permutation such that 
	\begin{align*}
	\frac{x_{\sigma_x^w(i)}}{w_{\sigma_x^w(i)}} \leq \frac{x_{\sigma_x^w(i+1)}}{w_{\sigma_x^w(i+1)}}
	\end{align*}
	holds for all $i\in \{1,\dots, n-1\}$. The \emph{$w$-Lorenz-curve} (or \emph{$L^w$-curve}) for the vector~$x$ is given by the piece-wise linear function $L^w_x: [0, W]\mapsto \R$ with 
	\begin{align*}
	L^w_x(0) = 0, \qquad \text{and } L^w_x\left(\sum_{k=1}^i w_{\sigma_x^w(k)}\right) = \sum_{k=1}^i x_{\sigma_x^w(k)} 
	\end{align*}
	for all $i\in\{1,\dots, n-1\}$ such that $L^w_x$ is linear on each interval of the form $[0, w_{\sigma^w_x(1)}]$, $(\sum_{k=1}^i w_{\sigma^w_x(k)}, \sum_{k=1}^{i+1} w_{\sigma^w_x(k)}]$ for all $i\in\{1,\dots, n-1\}$. 
	
	For two vectors~$x,y$ in $\R^n$ we say \emph{$x$ $w$-Lorenz-dominates y} (or \emph{$x$ $L^w$-dominates y}) if $L^w_x(p) \geq L^w_y(p)$ for all $p\in [0,W]$ and strict inequality for at least one~$p\in [0,W]$.
\end{definition}

Note that we obtain a more tractable form of the $L^w$-curve for a given vector~$x$, $p\in[0,W]$ and $k\in \{1,\dots, n\}$. For $p\in [0, w_{\sigma^w_x(1)}]$, 
\begin{align}
L^w_x(p)&= p\frac{x_{\sigma_x^w(1)}}{w_{\sigma_x^w(1)}}
\intertext{and for $p\in(\sum_{k=1}^i w_{\sigma^w_x(k)}, \sum_{k=1}^{i+1} w_{\sigma^w_x(k)}]$, $i\in\{1,\dots, n-1\}$,}
L^w_x(p) = \sum_{k=1}^i x_{\sigma_x^w(k)} &+ \left(p-\sum_{k=1}^{i} w_{\sigma_x^w(k)}\right)\cdot \frac{x_{\sigma_x^w(k+1)}}{w_{\sigma_x^w(k+1)}}.\label{eq: lorenz interval k}
\end{align}

\begin{definition}[weighted egalitarian allocation, Lorenz-core]
	For a set~$M\in\R^n$ for $n\in\N$ and weights~$w\in\R^n_{>0}$, we define the set of vectors in $M$ that are $w$-Lorenz-undominated within $M$ as
	\begin{align*}
	EA^w(M) \coloneqq \{x\in M : \not\exists y\in M, y\neq x, \text{ s.t. } y \; L^w\text{-dominates } x\}. 
	\end{align*}
	The \emph{Lorenz-core}~$L(S)$ of a coalition~$S$ is now defined recursively. First, the Lorenz-core of a singleton coalition~$\{i\}$ is given by $L(\{i\}) = \{v(i)\}$. Then, for a coalition~$S$, the Lorenz-core is defined as
	\begin{align*}
	L(S) \coloneqq \{x\in\R^{|S|} : \sum_{i\in S} x_i = v(S), \not\exists T\subseteq S \text{ and } y\in EA^w(L(T)) \text{ s.t. } y > x_{|T}\}, 
	\end{align*}
	where $x_{|T}$ denotes the projection of $x$ on $T$ and the comparison is done entry-wise. Finally, the set of \emph{weighted egalitarian allocations} is denoted by $EA^w(L(N))$. 
\end{definition}

Koster \cite{koster1999weighted} adapted an algorithm of Dutta and Ray \cite{dutta1989concept} to compute the weighted egalitarian allocation vector. Algorithm~\ref{alg: egalitarian allocation} iteratively computes a partition of the participant set~$N$ into coalitions and thus terminates after at most $|N|$ steps. In each step a coalition~$S$ with largest average value~$a(S,v) = \frac{v(S)}{w(S)}$ has to be found, which can be computationally hard.

\begin{algorithm}
	\caption{Compute Egalitarian Allocation (cf. \cite{dutta1989concept, koster1999weighted})}\label{alg: egalitarian allocation}
	Input: A TU-game~$G=(N,v)$, weights~$w\in\R^{|N|}_{>0}$\\
	Initialization: $v_1 = v$, $N_1 \leftarrow N$, $x\leftarrow 0$, $i\leftarrow 1$\\
	\While{$N\neq\emptyset$}{
		Compute coalition~$S_i\subseteq N$ with maximal weighted average value~$a_w(S_i, v_i) = \frac{v_i(S)}{w(S)}$. \\
		\ForEach{$p\in S_i$}{$x_p = w_pa_w(S_i,v_i)$.}
		$N \leftarrow N\backslash S_i$ \\
		Define $v_{i+1}$ by $v_{i+1}(S) \leftarrow v_i(S\cup S_i) - v_i(S_i)$ for all $S\subseteq N$.
	}
	Return $x$.
\end{algorithm}

A slightly different concept of the Lorenz-core was introduced e.g. by Arin et al. \cite{arin1997egalitarianism}, where the Lorenz-core is defined as the Lorenz-undominated payoff vectors in the core. Clearly, by restricting to core elements only, they also restrict the possible egalitarian allocations. The next theorem is due to Koster \cite{koster1999weighted} and a generalization of a result by Dutta et al. \cite{dutta1989concept}, which shows that the restriction to core elements has no effect for convex games (or, more general, balanced games) since the core of such games is always non-empty. 

\begin{theorem}[cf. \cite{dutta1989concept, koster1999weighted}]\label{thm: Lorenz dominates core}
	Let~$(N,v)$ be a convex game and $w$ a vector of positive weights in $\R^n_{>0}$. It holds true that there exists a unique weighted egalitarian allocation~$y^*$ with weights~$w$. Furthermore, $y^*$ is a core allocation and $w$-Lorenz-dominates every other core allocation. 
\end{theorem}

\begin{example}
	Let the characteristic function~$v$ be given by $v(\{i\}) = 0$ for $i=1,2,3$, and $v(1,2) = 5$, and $v(1,3) = 0$, and $v(2,3) = 2$ and $v(1,2,3) = 7$. Then the Vickrey payments are given by $vp_1 = 5$, $vp_2 = 7$ and $vp_3 = 2$. Since the sum exceeds the available budget, we consider the (equally, individually) scaled Vickrey payments. These are given by $svp_1 = 2.5$, $svp_2 = 3.5$ and $svp_3 = 1$. In this case the equally and individually scaled Vickrey payments coincide. 
	
	Now we consider the egalitarian allocation with unit weights first. Thus, the coalition with the highest average value is $(1,2)$ and thus $ea_1 = ea_2 = 2.5$. With the updated characteristic function as described in Algorithm~\ref{alg: egalitarian allocation}, we get $ea_3 = 2$. 
	
	If we consider the Vickrey payments as weights for the egalitarian allocation, we see that the grand coalition~$(1,2,3)$ yields the highest average value~$1/2$. Thus, the weighted egalitarian allocation is given by $wea_1 = 2.5$, $wea_2 = 3.5$ and $wea_3 = 1$. 
\end{example}

\section{The Connection Between the ISV-Payments and the Egalitarian Allocation}\label{sec: egal: connection}
In this section we show that the weighted egalitarian allocation for specific weights and the ISV-payments coincide, i.e. that the ISV-payments can be seen as a special case of the weighted egalitarian allocation. Given two vectors~$x,y\in\R^n$ and weights~$w\in\R^n_{>0}$, the following two lemmas give us useful implications if $x$ is $L^w$-dominated by $y$.
\begin{lemma}\label{lem: Lorenz dominates one entry strictly bigger}
	Let $x,y$ be two vectors in~$\R^n$ and let $w$ denote weights in $\R^n_{>0}$. If~$x$ is $L^w$-dominated by~$y$, then there exists an index~$j\in \{1,\dots, n\}$ such that~$y_j>x_j$ holds true.  	
\end{lemma}
\begin{proof}
	Suppose for contradiction that $y_j\leq x_j$ holds true for all~$j=1,\dots,n$. Then also the ordered vectors~$x_{\sigma_x^w}$, $y_{\sigma_y^w}$ fulfill $x_{\sigma_x^w(j)} \leq y_{\sigma_y^w(j)}$ for all $j=1,\dots,n$. Since $x$ is $L^w$-dominated by $y$, there exists a $p\in[0,W]$ with $L^w_y(p) > L^w_x(p)$. Now choose $k\in \{1,\dots, n\}$ such that condition~\eqref{eq: lorenz interval k} is fulfilled. Then it holds
	\begin{align*}
	L^w_x(p) &= \sum_{u=1}^{k} x_{\sigma_x^w(u)} + \left(p - \sum_{u=1}^{k}w_{\sigma_x^w(u)}\right)\frac{x_{\sigma_x^w(k+1)}}{w_{\sigma_x^w(k+1)}}
	\intertext{By our assumption, this can be estimated downwards by}
	\geq&\sum_{u=1}^{k} y_{\sigma_y^w(u)} + \left(p - \sum_{u=1}^{k}w_{\sigma_x^w(u)}\right)\frac{y_{\sigma_y^w(k+1)}}{w_{\sigma_x^w(k+1)}}\\
	\geq& \sum_{u=1}^{k} y_{\sigma_y^w(u)} + \left(p - \sum_{u=1}^{k}w_{\sigma_y^w(u)}\right)\frac{y_{\sigma_y^w(k+1)}}{w_{\sigma_y^w(k+1)}}\\
	=& L^w_y(p),
	\end{align*}
	by the definition of the Lorenz-curve and the permutation~$\sigma_y^w$. This is a contradiction to the assumption. 
\end{proof}

\begin{lemma}\label{lem: dominates minimum}
	Let $x,y$ be two vectors in~$\R^n$ and weights~$w \in\R^n_{>0}$. Further, let $y$ $L^w$-dominate $x$. Then there exists an index~$j\in \{1,\dots, n\}$ such that
	\begin{align}
	y_j > x_j & \qquad \text{and } \qquad \frac{y_i}{w_i} \geq \min\left\{\frac{x_i}{w_i}, \frac{x_j}{w_j}\right\} \text{ for all } i\in \{1,\dots, n\}\label{eq: yi/wi geq min xiwi..}
	\end{align}
	holds true. 
\end{lemma}
\begin{proof} 
	Since $x$ is $L^w$-dominated by~$y$, we get 
	\begin{align}
	\sum_{u=1}^k x_{\sigma_x^w(u)} & = L^w_x\left(\sum_{u=1}^k w_{\sigma_x^w(u)}\right) \notag \\
	& \leq L^w_y\left(\sum_{u=1}^k w_{\sigma_x^w(u)}\right) \notag \\
	& \leq \sum_{u=1}^k y_{\sigma_y^w(u)} + \left(\sum_{u=1}^{k}w_{\sigma_x^w(u)} - \sum_{u=1}^{k}w_{\sigma_y^w(u)}\right) \frac{y_{\sigma_y^w(k+1)}}{w_{\sigma_y^w(k+1)}} \notag \\
	& \leq \sum_{u=1}^k y_{\sigma_x^w(u)} + \left(\sum_{u=1}^{k}w_{\sigma_x^w(u)} - \sum_{u=1}^{k}w_{\sigma_x^w(u)}\right) \frac{y_{\sigma_y^w(k+1)}}{w_{\sigma_x^w(k+1)}} \notag \\
	& \leq \sum_{u=1}^ky_{\sigma_x^w(u)}\label{eq: inequality Lorenz curve}
	\end{align}
	for all~$k\in \{1,\dots, n\}$ by the definition of $\sigma^w_x$ and $\sigma^w_y$. Also we know by Lemma~\ref{lem: Lorenz dominates one entry strictly bigger} that at least one entry of~$y$ is strictly larger than the corresponding entry of~$x$. Thus, the inequality~\eqref{eq: inequality Lorenz curve} is strict for some~$k\in \{1,\dots, n\}$. Let~$\overline{k}$ be the smallest index such that~\eqref{eq: inequality Lorenz curve} is strict and, thus, we get
	\begin{align*}
	y_{\sigma_x^w(\overline{k})} &> x_{\sigma_x^w(\overline{k})}
	\intertext{and}
	y_{\sigma_x^w(k)} &= x_{\sigma_x^w(k)} \quad \text{ for all } k<\overline{k}.
	\end{align*}
	
	Now let $j\coloneqq \sigma_x^w(\overline{k})$ which implies $y_j > x_j$. Suppose for contradiction that there exists an index~$i\in \{1,\dots, n\}$ such that~\eqref{eq: yi/wi geq min xiwi..} is not fulfilled, i.e. $\frac{y_i}{w_i} < \min \left\{\frac{x_i}{w_i}, \frac{x_j}{w_j}\right\}$. Hence, $y_i < x_i$ and $\frac{y_i}{w_i} < \frac{x_i}{w_i}$ hold. Since the former inequality is true, the index~$\overline{u}$ with $\sigma_x(\overline{u}) = i$ is greater or equal to~$\overline{k}$. We define
	\begin{align}
	\varepsilon&\coloneqq\min\left(w_i, w_{\sigma_x(\overline{k})}\right), \\
	p&\coloneqq \sum_{u=1}^{\overline{k}-1} w_{\sigma_x^w(u)} + \varepsilon.
	\end{align}
	
	By evaluating the Lorenz curve~$L^w_y$ at point~$p$ we get 
	\begin{align*}
	L^w_y(p) &= L^w_y\left(\sum_{u=1}^{\overline{k} - 1} w_{\sigma_x^w(u)} + \varepsilon\right)\leq \sum_{u=1}^{\overline{k}-1} y_{\sigma_x^w(u)} + \varepsilon\frac{y_i}{w_i}. 
	\end{align*} 
	By the definition of~$\overline{k}$, this sum is equal to
	\begin{align*}
	\sum_{u=1}^{\overline{k}-1} x_{\sigma_x^w(u)} + \varepsilon\frac{y_i}{w_i}.
	\end{align*}
	Using that $\frac{y_i}{w_i} < \min\left\{\frac{x_i}{w_i}\frac{x_j}{w_j}\right\}$ and the definition of~$j$, we get
	\begin{align*}
	\sum_{u=1}^{\overline{k}-1} x_{\sigma_x^w(u)} + \varepsilon\frac{y_i}{w_i} &< \sum_{u=1}^{\overline{k}-1} x_{\sigma_x^w(u)} + \varepsilon\frac{x_j}{w_j} \\
	&= \sum_{u=1}^{\overline{k}-1} x_{\sigma_x^w(u)} + \varepsilon\frac{x_{\sigma_x^w(\overline{k})}}{w_{\sigma_x^w(\overline{k})}} \\
	&= L^w_x\left(\sum_{u=1}^{\overline{k}-1}w_{\sigma_x^w(u)} + \varepsilon\right) \\
	&= L^w_x(p),
	\end{align*}
	which is a contradiction. This settles the claim.
\end{proof}

Given the condition of Lemma~\ref{lem: dominates minimum}, we show in the next lemma that we can construct a lexicographically bigger vector.

\begin{lemma}\label{lem: lexicographically dominating}
	Let~$x,y$ be two vectors in~$\R^n$ and an index~$j\in \{1,\dots, n\}$ such that~$y_j>x_j$ and 
	\begin{align}
	y_i \geq\min\{x_i,x_j \} \text{ for all } i\in \{1,\dots, n\} \label{eq:yi greater than min}
	\end{align}
	holds true. Then $x$ is lexicographically dominated by the vector~$z\coloneqq\frac{x+y}{2}$.	
\end{lemma}
\begin{proof}
	Let $k$ be the number of entries in $x$ that are less than or equal to $x_j$. Since the permutation~$\sigma_x$ sorts the entries of~$x$ in a non-decreasingly way, it holds that $x_{\sigma_x(k)} = x_j$. Notice that this does not imply $\sigma_x(k) = j$. We distinguish two cases. Either we consider an index~$i$ out of
	\begin{align*}
	L_j\coloneqq\{l: z_{\sigma_z(l)}\leq x_{\sigma_x(l)} = x_j \}
	\intertext{or out of} 
	R_j\coloneqq\{l: z_{\sigma_z(l)}> x_{\sigma_x(l)} = x_j \}.
	\end{align*}
	
	Given $i\in L_j$, we know that $z_{\sigma_z(i)}\leq x_j$ holds. Further, this implies that one of the following cases
	\begin{align}
	x_i < x_j \quad\text{ or}  \label{eq: xi less than xj}\\
	y_i < x_j \quad\text{ or} \label{eq: y less than x} \\
	x_i = y_i = x_j \label{eq: all entries are equal}
	\end{align}
	holds. In case~\eqref{eq: all entries are equal} and case~\eqref{eq: xi less than xj} it follows directly that $x_i \leq x_j$. 
	With~\eqref{eq:yi greater than min} and the fact that $\min\{x_i,x_j\} = x_i$ holds true, it follows also for case~\eqref{eq: y less than x} that $x_i \leq x_j$ holds true. With this, we obtain
	\begin{align*}
	\{i: z_i\leq x_j\} \subseteq \{i: x_i\leq x_j\}
	\end{align*}
	and since $z_j > \frac{x_j + x_j}{2} = x_j$ by definition, this inclusion is strict. This means, the index~$j$ is only contained in the latter set and we get an upper bound on the size of~$\{i: z_i\leq x_j\}$ by~$k$. Thus, $z$ has strictly less than $k$~entries smaller than~$x_j$ and, hence, also for a permutation~$\sigma$ that orders the entries of a vector non-decreasingly, it holds
	\begin{align}
	z_{\sigma_z(k)} > x_j = x_{\sigma_x(k)}. \label{eq: z greater than x for k}
	\end{align}
	
	In the same way, for $u\in\{1,\dots, k\}$ and each $i\in \{1,\dots, n\}$ with $z_i\leq x_{\sigma_x(u)}$ we obtain that $z_i < x_{\sigma_x(u)}$ implies  $x_i < x_{\sigma_x(u)}$.
	Thus, we get again an inclusion
	\begin{align}
	\{i: z_i < x_{\sigma_x(k)}\} \subseteq \{i: x_i < x_{\sigma_x(u)}\}
	\end{align}
	which is not necessarily strict. Since $x_{\sigma_x(u)}$ is the $u$-th smallest entry of~$x$, we get an upper bound on the size of the set~$\{i: z_i < x_{\sigma_x(k)}\}$ by~$u$. This implies 
	\begin{align}
	z_{\sigma_z(u)} \geq x_{\sigma_x(u)} \text{ for all } u\in\{1,\dots, k\}.\label{eq: z greater than x for u}
	\end{align}	
	With \eqref{eq: z greater than x for k} and \eqref{eq: z greater than x for u} it follows that $x$ is lexicographically dominated by~$z$.
	
	We do not have to consider case~\ref{eq: y less than x} since~$x$ is lexicographically dominated by $z$ in any way.  
\end{proof}

Considering the egalitarian allocation, i.e. the weighted egalitarian allocation with unit weights, we are able to obtain the next lemma.

\begin{lemma}\label{cor: egal alloc is lex max}
	Let $y\in\R^n$ be the egalitarian allocation. Then $y$ is the lexicographical maximal vector in the core.
\end{lemma}
\begin{proof}
	Suppose not. Then there exists a vector~$x\in\R^n$ with $y\neq x$ which is the lexicographical maximal vector in the core. By Lemma~\ref{lem: dominates minimum}, there exists an index~$j$ such that $y_j>_j$ and $y_i\geq \min\{x_i, x_j\}$ for all $i=1,\dots,n$. Therefore, the condition of Lemma~\ref{lem: lexicographically dominating} is fulfilled for the vectors~$x,y$. Thus, by said lemma, it follows that $z\coloneqq \frac{x+y}{2}$ is lexicographical bigger than $x$ --- a contradiction. 
\end{proof}

Together with the results above we are able to show that the ISV-payments coincide with the weighted egalitarian allocation if we take the Vickrey payments as weights. We summarize this in the following theorem. 

\begin{theorem}
	Let~$(N,v)$ be a convex game and let~$w\in\R_{>0}^{|N|}$ denote weights for the participants. Then the ISV-payments are equal to the weighted egalitarian allocation with the weights equal to the Vickrey payments. 
\end{theorem}
\begin{proof}
	Let~$x$ denote the ISV-payment vector and $y$ the egalitarian allocation vector with $x\neq y$. Note that the vector~$x$ lies in the core by construction and the egalitarian allocation lies in the core by Theorem~\ref{thm: Lorenz dominates core}. Further, also by Theorem~\ref{thm: Lorenz dominates core}, the allocation~$y$ $L^w$-dominates~$x$. 
	
	By Lemma~\ref{lem: dominates minimum} we know that there exists an index~$j$ with $y_j > x_j$ and $\frac{y_i}{w_i}\geq \min\{\frac{x_i}{w_i}, \frac{x_j}{w_j}\}$ for all $i=1,\dots, n$. Thus, we can apply Lemma~\ref{lem: lexicographically dominating} to the vectors~$\frac{y}{w}$ and $\frac{x}{w}$ and get that the vector~$z=\frac{\frac{x+y}{2}}{w}$ lexicographically dominates the vector~$\frac{x}{w}$. Since the core is a convex set, also $z'\coloneqq \frac{x+y}{2}$ lies in the core. Since the weight~$w$ is given as the Vickrey payments, $\frac{x}{w}$ is by definition the lexicographic maximal solution to Problem~\eqref{eq: obj func isv}. Since the Vickrey payments is the highest payoff to a participant of all elements in the core, the vector $\frac{z'}{w}$ is also a solution to Problem~\eqref{eq: obj func isv}. This is a contradiction to $x$ being the ISV-payment vector.
\end{proof}

We conclude the section on the connection between the ISV payments and the weighted egalitarian allocation with the following corollaries which follow immediately from the result above. 
\begin{corollary}
	The ISV-payments are a unique core element (cf. \cite{dutta1989concept, koster1999weighted}). $\hfill \Box$
\end{corollary}

\begin{corollary}
	If the Vickrey payments for each participant can be computed efficiently, also the corresponding weighted egalitarian solution can be computed in polynomial time by iteratively solving at most $|N|$ linear programs (cf. \cite{ewe2011}). $\hfill \Box$
\end{corollary}

\section{Conclusion}

In this paper we exhibit a connection between the ISV payments and the WEA. We gave a short recall of the Vickrey payments as well as the EA and presented introductory examples. The main result here is the connection between the ISV payments and the WEA, where the weights are given by the Vickrey payments. This allows to conclude the uniqueness of the ISV payments as well as a polynominal time algorithm for computing the WEA with said weights. 

\nocite{graef2019}
\bibliographystyle{plain}
\bibliography{references}

\newpage
\noindent
Till Heller\\
Department of Optimization\\
Fraunhofer ITWM, Kaiserslautern\\
Germany\\
ORCiD: 0000-0002-8227-9353\\

Niklas Gräf\\
Sven O. Krumke\\
Optimization Research Group, Department of Mathematics\\
Technische Universit\"at Kaiserslautern, Kaiserslautern\\
Germany\\
\end{document}